\documentclass[a4paper,UKenglish,cleveref, autoref, thm-restate]{lipics-v2021}

\title{Fine-Grained Complexity of Multiple Domination and Dominating Patterns in Sparse Graphs}

\author{Marvin Künnemann}
{Karlsruhe Institute of Technology, Germany}
{}
{}
{}
\author{Mirza Redzic}
{Karlsruhe Institute of Technology, Germany}
{}
{}
{}

\authorrunning{M. Künnemann, M. Redzic} 

\Copyright{Marvin Künnemann and Mirza Redzic} 



\category{} 

\relatedversion{} 

\acknowledgements{Research supported by the Deutsche Forschungsgemeinschaft (DFG, German 
Research
Foundation) – 462679611.}

\nolinenumbers 

\usepackage[mathscr]{eucal}
\usepackage{subcaption}

\newcommand{\bigO}{\mathcal{O}}

\usepackage{comment,color}
\usepackage{mathtools}

\usepackage{cancel}
\usepackage[linesnumbered,boxed]{algorithm2e}

\usepackage{enumitem}
\DeclareMathOperator\MM{MM}

\SetCommentSty{mycommfont}

\newcommand{\rem}[3]{\textcolor{blue}{\textsc{#1 #2:}}
  \textcolor{red}{\textsf{#3}}}

\newcommand{\mirza}[2][says]{\rem{Mirza}{#1}{#2}}

\newcommand{\tOh}{\tilde{\mathcal{O}}}

\newcommand{\FOP}[1]{$\mathrm{FOP}_k$}

\makeatletter
\renewcommand\paragraph{%
  \@startsection{paragraph}
    {4}
    {\z@}
    {3.25ex \@plus1ex \@minus.2ex}
    {-1em}
    {\normalfont\normalsize\bfseries\addperiod}}
\newcommand{\addperiod}[1]{#1\@addpunct{.}}
\usepackage{textgreek}
\usepackage{apxproof}

\begin{document}

\maketitle

\begin{abstract}
The study of domination in graphs has led to a variety of \emph{dominating set} problems studied in the literature. Most of these follow the following general framework: Given a graph $G$ and an integer~$k$, decide if there is a set $S$ of $k$ vertices such that (1) some inner connectivity property $\phi(S)$ (e.g., connectedness) is satisfied, and (2) each vertex $v$ satisfies some domination property $\rho(S, v)$ (e.g., there is some $s\in S$ that is adjacent to $v$).

Since many real-world graphs are \emph{sparse}, we seek to determine the optimal running time of such problems in both the number $n$ of vertices and the number $m$ of edges in $G$. While the classic dominating set problem admits a rather  limited improvement in sparse graphs (Fischer, Künnemann, Redzic SODA'24), we show that natural variants studied in the literature admit much larger speed-ups, with a diverse set of possible running times. Specifically, if the matrix exponent $\omega$ equals $2$, we obtain conditionally optimal algorithms for:   
\begin{itemize}
 \item \emph{$r$-Multiple $k$-Dominating Set} (each vertex $v$ must be adjacent to at least $r$ vertices in $S$): If $r\le k-2$, we obtain a running time of $(m/n)^{r} n^{k-r+o(1)}$ that is conditionally optimal assuming the 3-uniform hyperclique hypothesis. In sparse graphs, this fully interpolates between $n^{k-1\pm o(1)}$ and $n^{2\pm o(1)}$, depending on $r$. Curiously, when $r=k-1$, we obtain a randomized algorithm beating $(m/n)^{k-1} n^{1+o(1)}$ and we show that this algorithm is close to optimal under the $k$-clique hypothesis.
 \item \emph{$H$-Dominating Set} ($S$ must induce a pattern $H$). We conditionally settle the complexity of three such problems: (a) Dominating Clique ($H$ is a $k$-clique), (b) Maximal Independent Set of size~$k$ ($H$ is an independent set on $k$ vertices), (c) Dominating Induced Matching ($H$ is a perfect matching on $k$ vertices). For all sufficiently large $k$, we provide algorithms with running time $(m/n)m^{(k-1)/2+o(1)}$ for (a) and (b), and $m^{k/2+o(1)}$ for (c). We show that these algorithms are essentially optimal under the $k$-Orthogonal Vectors Hypothesis ($k$-OVH). This is in contrast to $H$ being the $k$-Star, which is susceptible only to a very limited improvement, with the best algorithm running in time $n^{k-1 \pm o(1)}$ in sparse graphs under $k$-OVH. 
 \end{itemize}
\end{abstract}
\newpage

\section{Introduction}
\label{sec:introduction}
Domination in graphs is among the central topics in graph theory.
Although the earliest evidence of interest in concepts related to domination can be traced back to the mid 1800s in connection with various chess problems, it was introduced only a century later, in 1958, as a graph-theoretical concept by Claude Berge.
It has since gained a lot of attention and has been well-studied from both a graph-theoretic perspective, e.g.,~\cite{AlvaradoDMR19,BianchiNTT21,CutlerR16,SuzukiMN16,Taletskii22}, and an algorithmic perspective, e.g.,~\cite{eisenbrand2004complexity,PatrascuW10,FominLST12,FominT03,SLM18,SolomonU23}.
This problem has also played a central role in the field of complexity theory. 
Besides being one of the classic NP-complete problems, the $k$-Dominating Set problem has proven valuable within the realm of parameterized complexity theory, where it is regarded as perhaps the most natural $W[2]$-complete problem~\cite{DowneyF95}, as well as fine-grained complexity in P, where it was among first problems for which tight lower bounds under the Strong Exponential Time Hypothesis (SETH) have been established~\cite{PatrascuW10}.

Over the years, the concept of domination in graphs has spawned many natural variations, each offering unique insights into the structural properties of a graph, as well as different practical applications (e.g. in analysing sensor networks, facility management, studying influence in social networks, etc.). 
Some examples of such variations include total domination, paired domination, independent domination, multiple domination, etc.
Most of these domination problems satisfy the following framework:
We are given a graph $G$ and an integer $k$ and want to decide if there exists a set of vertices $S=\{x_1,\dots, x_k\}$ that satisfies some fixed inner property $\phi(x_1,\dots, x_k)$ such that for every vertex $v\in V(G)$ the \emph{domination} property $\rho(x_1,\dots, x_k, v)$ is satisfied.
Some examples of inner properties $\phi$ include:
\begin{itemize}
    \item $x_1,\dots, x_k$ are connected (Connected Domination).
    \item $x_1, \dots, x_k$ form an independent set (Independent Domination).
    \item Each $x_i\in S$ is adjacent to at least one $x_j\in S\setminus\{x_i\}$ (Total Domination).
\end{itemize}
Examples of the domination property $\rho$ include:
\begin{itemize}
    \item There exists $x_i\in S$ such that $d(x_i,v)\leq r$ (Domination at Distance $r$).
    \item $v$ is adjacent to at least $r$ distinct vertices $x_{i_1},\dots, x_{i_r}\in S$ ($r$-Multiple Domination).
    \item There exists a path of length $r$ between $v$ and some $x_i\in S$ ($r$-Step Domination).
\end{itemize}
Many of these domination problems have not seen any polynomial improvements over brute force in dense graphs, i.e., the best known algorithms for finding a solution of size~$k$ typically run in $\Omega(n^{k})$ time and for some variants it has been shown that improving upon these algorithms significantly would refute some of the popular fine-grained complexity assumptions.
Most notably, Pătraşcu and Williams~\cite{PatrascuW10} show that an $\bigO(n^{k-\epsilon})$ algorithm solving $k$-Dominating Set, for any~\makebox{$k\ge 3$} and $\epsilon > 0$ would refute the Strong Exponential Time Hypothesis (SETH).
However, by far not all graphs of interest are dense. 
Particularly, many real-world graphs, for which the domination problems have been extensively used, are typically sparse (e.g. social networks, sensor networks, road networks, etc.).
Hence, it is natural to ask what is the best running time of domination problems in sparse graphs.
Recently, Fischer, Künnemann and Redzic~\cite{fischer2024effect} proved that the fine-grained complexity of $k$-Dominating set shows a non-trivial sensitivity to sparsity of the input graph. 
More precisely, despite the SETH-based lower bound of $n^{k-o(1)}$, they prove that when the input graph is sufficiently sparse, we can in fact improve upon this running time significantly by using sparse matrix multiplication techniques, and obtain a conditionally optimal running time of $mn^{k-2+o(1)}$ for all $k\ge 2$ (assuming $\omega=2$).
This raises the question if we can obtain similar improvements in sparse graphs for other natural domination problems.
In this paper we consider two natural classes of domination problems that exhibit an interesting sensitivity to sparsity, namely \emph{$r$-Multiple Domination} and \emph{Pattern Domination}.
\subparagraph{Multiple Domination in Graphs.} The concept of Multiple Domination has been introduced as a generalization of Dominating Set by Fink and Jacobson in 1985~\cite{fink1985n,finkj1985n} and has been intensively studied since (see e.g.~\cite{Araki08,Araki09,argiroffo2018complexity,argiroffo2015complexity,GagarinZ08,HenningK10,LiWS23}). 
For a graph $G=(V,E)$ we say a subset of vertices $S$ is an \emph{$r$-multiple dominating set} if each vertex $v\in V\setminus S$ has at least $r$ neighbours in $S$.\footnote{We remark that in the literature, this concept is better known under the name \emph{$r$-Dominating Set}. In the setting of parameterized complexity, however, the notion of $k$-Dominating Set usually refers to dominating sets of size $k$, so for clarity, we use the term \emph{$r$-Multiple Dominating Set}.}
Given a graph $G$ with $n$ vertices and $m$ edges, the \emph{$r$-Multiple $k$-Dominating Set} problem is to decide if there is an $r$-multiple dominating set $S$ of size at most $k$.
Harary and Haynes~\cite{HARARY199699,HararyH00} introduced, in two papers published in 1996 and 2000, a very related concept of double domination and, more generally, the \emph{$r$-Tuple Dominating Set}, which is a subset of vertices $S$, such that the closed neighborhood of every vertex $v\in V$ intersects with $S$ in at least $r$ elements.
We note that all of the algorithms and lower bounds that we provide for $r$-Multiple $k$-Dominating Set work with very minor modifications for $r$-Tuple Dominating Set as well.

We aim to settle the fine grained complexity of this problem in sparse graphs.
Interestingly, the hardness of this problem depends not only on the trade-off between $m$ and $n$, but also on the trade-off between $r$ and $k$.
In particular, we distinguish between the two cases, $r\leq k-2$ and $r=k-1$ (note that if $r\geq k$, the problem becomes trivial) -- it turns out that the fine-grained complexity differs already in the dense case for these two cases.

Let us begin with a baseline algorithm for the case $r\leq k-2$: 
\begin{theorem}\label{th:rleqk-2:algorithm-simple}
Let $k\geq 3$ and $r\leq k-2$ be fixed constants. Given a graph $G$ with $n$ vertices and $m$ edges we can solve $r$-Multiple $k$-Dominating Set in time $(m/n)^r n^{k-r+o(1)}$, assuming $\omega = 2$.
\end{theorem}
We note the remarkable improvement by a factor of $\Theta(\frac{n^{2r}}{m^r})$ over the best known algorithm in dense graphs.
Already for Double $k$-Dominating Set (i.e., $r=2$), this yields an algorithm running in $m^2n^{k-4 + o(1)}$, which beats the running time of the best $k$-Dominating Set algorithm \cite{fischer2024effect} by a factor of $\Theta(n^2/m)$.
Even better, if we want each vertex in our graph to be dominated by precisely $50\%$ of the solution vertices, we get an algorithm running in $m^{\frac{k}{2}+o(1)}$, halving the exponent in sparse graphs.
Perhaps surprisingly, for $r$ close to $k$, this yields an algorithm whose running time exponent is independent of $k$ when the input graph is very sparse.
In particular, for $r=k-2$, this running time becomes essentially quadratic in very sparse graphs ($m=\Tilde\bigO(n)$).
The question remains whether this running time is best possible -- perhaps we can always obtain $m^{\frac{k}{2} + o(1)}$ (or even better) running time when $r\geq 2$?
We answer this question negatively, and in fact show that any polynomial improvement over our algorithm (up to resolving the matrix multiplication exponent $\omega$) would refute the $3$-uniform Hyperclique Hypothesis\footnote{For a definition of the $3$-uniform Hyperclique Hypothesis, we refer to Section \ref{app-hardness}.}, thus settling the complexity of this problem in sparse graphs whenever $r\leq k-2$.
\begin{restatable}{theorem}{TheoremrSmallLB}\label{th:rleqk-2:lb}
    Let $k\geq 3$, $r\leq k-2$ be fixed constants, and $\varepsilon>0$.
    An algorithm solving $r$-Multiple $k$-Dominating Set in time $O((m/n)^{r}n^{k-r-\varepsilon})$ would refute $3$-uniform Hyperclique Hypothesis. This holds even when restricting $m=\Theta(n^{1+\gamma})$ for any $0< \gamma \le 1$.
\end{restatable}
While the algorithmic approach of Theorem~\ref{th:rleqk-2:algorithm-simple} is applicable also for the remaining case of $r=k-1$, it turns out that the resulting upper bound of $(m/n)^{k-1} n^{1+o(1)}$ (if $\omega=2$), is not optimal in general: In fact, we reduce the problem to Clique detection, by observing that each pair of solution vertices dominates the whole graph (i.e., forms a dominating set of size 2). The resulting algorithm substantially improves over exhaustive search already in dense graphs.  
Furthermore, in sparse graphs, we can apply a Bloom-filter inspired randomized algorithm of~\cite{fischer2024effect}, allowing us to list all 2-dominating sets efficiently, to obtain an efficient randomized reduction to an Unbalanced $k$-Clique Detection instance with $k-1$ parts of size $\bigO(\frac{m}{n})$ and one of size $n$, which we denote as $k$-Clique($\frac{m}{n},\dots, \frac{m}{n}, n$).\footnote{For more details, we refer to Section~\ref{sec:prelims} and Section~\ref{sec:multiple-kds}.}
\begin{restatable}{theorem}{TheoremReductionToUnbalancedClique}\label{theorem:reduction-to-unbalanced-clique}
    For any fixed constant $k\geq 2$ let $T_k(m, n)$ denote the time required to solve the $k$-Clique($\frac{m}{n},\dots, \frac{m}{n}, n$) problem. There is a randomized algorithm solving $(k-1)$-Multiple $k$-Dominating Set in time $m^{\frac{\omega}{2}+o(1)} + \bigO(T_k(m,n))$.
\end{restatable}
It remains to analyze the complexity of the Unbalanced Clique problem.
It is straightforward to obtain an algorithm solving this problem in time $\frac{m^2}{n} \cdot \big(\frac{m}{n}\big)^{\omega k/3+o(1)}$, which in case of very sparse graphs ($m=\tilde \bigO(n)$) yields a near-linear running time.
However, this running time analysis is still crude, and we can do even better for sufficiently ``nice'' values of $m/n$ and $k$.
More precisely, we show that for each positive integer $p$, we can solve $k$-Clique($n^{\frac{1}{p}}, \dots, n^{\frac{1}{p}}, n$) in time $(n^{\frac{1}{p}(k-1)+1})^{\frac{\omega}{3}+o(1)}$ for all sufficiently large $k$ satisfying $k\equiv 2p+1 \pmod 3$.
\begin{restatable}{proposition}{proposition:unbalanced-k-clique-running-time}
    Let $G=(V_1,\dots, V_k, E)$ be a $k$-partite graph with vertex part sizes $|V_1| = n$ and ${|V_2|=\dots = |V_{k}| = n^\gamma}$ for some $0\leq \gamma\leq 1$. 
    If 
    \begin{enumerate}
        \item $(k-1 + \frac{1}{\gamma})$ is an integer divisible by $3$,\label{item:divisibility-by-3}
        \item $\frac{2}{\gamma}<k-1$, \label{item:bounding-gamma-k}
    \end{enumerate}
    then we can decide if $G$ has a clique of size $k$ in time $(n^{\gamma(k-1)+1})^{\frac{\omega}{3}+o(1)}$.
\end{restatable}
Applying this running time to our setting, this yields an algorithm solving $(k-1)$-Multiple $k$-Dominating Set in the applicable cases in time $\big(n(\frac{m}{n})^{k-1}\big)^{\omega/3 + o(1)}$.
We also complement this with a matching conditional lower bound based on the $k$-Clique Hypothesis for the $(k-1)$-Multiple $k$-Dominating Set problem.
\begin{restatable}{theorem}{theorem:reqk-1:lb}
    Let $G$ be a graph with $n$ vertices and $m=\Theta(n^{1+\gamma})$ edges for any rational $0\leq \gamma\leq 1$. 
    For any $\varepsilon>0$, an algorithm solving $(k-1)$-Multiple $k$-Dominating Set on $G$ in time 
    $\big(n(\frac{m}{n})^{k-1}\big)^{\omega/3 - \varepsilon}$
    would refute the $k$-Clique Hypothesis.
\end{restatable}
\subparagraph{Pattern Domination in Graphs.} 
For graphs $G$ and $H$, we say a subset $S\subseteq V(G)$ is \emph{$H$-Dominating Set} if $S$ dominates $G$ and induces a subgraph of $G$ that is isomorphic to $H$.
For a fixed constant $k$, we define a \emph{$k$-Pattern Dominating Set Problem} as follows.
Given a graph $G$ with $n$ vertices and $m$ edges and a graph $H$ with $k$ vertices, decide if $G$ contains an $H$-Dominating Set.
We can observe that this problem is at most as hard as listing all dominating sets of size $k$ in a graph and hence we can solve it in $mn^{k-2+o(1)}$ for all sufficiently large $k$.
On the other hand, it has been implicitly proved in \cite{fischer2024effect} that detecting the patterns isomorphic to the star graph $K_{1,(k-1)}$ that dominate $G$ takes at least $mn^{k-2-o(1)}$, unless $k$-OV hypothesis fails, hence settling the fine-grained complexity of $k$-Pattern Domination in sparse graphs.
A more interesting direction is to ask what happens to the complexity of this problem if $H$ is a fixed graph rather than a part of the input.
We call this problem \emph{$H$-Dominating Set Problem}. 
In the context of graph theory, this class of problems has been widely studied, for a variety of natural choices of pattern $H$ that include:
\begin{itemize}
    \item Dominating Clique~\cite{chepoi1998note,cozzens1990dominating,dragan1994dominating,Kohler00,kratsch1990finding,Balliu0KO23,kratsch1994dominating}
    \item Dominating Independent Set~\cite{cho2023tight,cho2023independent,cho2023independentdom,kuenzel2023independent,pan2023improved}
    \item Dominating Path~\cite{FaudreeFGHJ18,FaudreeGJW17,Veldman83}
    \item Dominating Cycle~\cite{FANG2021112196,FANG202343}.
\end{itemize}
Notably, the Dominating Independent Set problem is equivalent to the well-known Maximal Independent Set problem \cite{AssadiK022,AssadiOSS18,AssadiOSS19,000122}. 
It turns out that the fine-grained complexity of $H$-Dominating Set problem in sparse graphs depends heavily on the choice of $H$.
Obviously, for any fixed $H$ this problem is at most as hard as the $k$-Pattern Dominating Set Problem and hence can be solved in time $mn^{k-2+o(1)}$.
As our first contribution for this problem we show that for no $k$-vertex graph $H$ and $\varepsilon>0$, can we solve this problem in the running time $\big(\frac{m^{(k+1)/2}}{n}\big)^{1-\varepsilon}$, unless $k$-OV hypothesis fails.
\begin{restatable}{theorem}{PropHDomLB}\label{prop:h-dom-lb}
    Let $H$ be any graph on $k\geq 3$ vertices.
    For no $\varepsilon>0$ is there an algorithm solving $H$-Dominating Set in time $O(\big(\frac{m^{(k+1)/2}}{n}\big)^{1-\varepsilon})$, unless the $k$-OV hypothesis fails.
\end{restatable}
We then consider the two most studied patterns, namely $k$-Clique and $k$-Independent Set.
For sufficiently large $k$, by leveraging the simple fact that there are at most $\bigO(m^{k/2})$ many cliques in the graph with $m$ edges as well as  fast matrix multiplication, we can obtain an algorithm for Dominating $k$-Clique problem running in time $\big(\frac{m^{(k+1)/2}}{n}\big)^{1+o(1)}$, thus matching the lower bound from \autoref{prop:h-dom-lb}.
\begin{restatable}{theorem}{cliqueDomAlg}
    Let $k\geq 5$ be a fixed constant. Dominating $k$-Clique problem can be solved on graphs with $n$ vertices and $m$ edges in time 
    \[
    \MM(\frac{m}{n}\cdot m^{\frac{1}{2}\lfloor\frac{k-1}{2}\rfloor}, n, m^{\frac{1}{2}\lceil\frac{k-1}{2}\rceil}).
    \]
    where $\MM(a,b,c)$ is the time required to multiply an $a\times b$ matrix by a $b\times c$ matrix. If $\omega = 2$, this becomes $\big (\frac{m^{(k+1)/2}}{n}\big)^{1+o(1)}$.
\end{restatable}
On the other hand, the number of $k$-Independent Sets in sparse graphs is typically much larger and can be as large as $\Theta(n^k)$.
Still, perhaps surprisingly, by leveraging some simple structural properties of maximal independent sets, we can obtain an algorithm matching the lower bound from \autoref{prop:h-dom-lb}.
\begin{restatable}{theorem}{indsetdomalg}\label{th:indset-dom-alg}
    Let $k\geq 3$ be fixed. Dominating $k$-Independent Set problem can be solved on graphs with $n$ vertices and $m$ edges in time $(\frac{m^{(k-1 + \omega)/2}}{n})^{1+o(1)}$.
    If $\omega = 2$, this becomes $\big (\frac{m^{(k+1)/2}}{n}\big)^{1+o(1)}$.
\end{restatable}
So far we mentioned the full classification of three structurally very different choices of patterns $H$, that all fall into one of the two extreme regimes of being either as hard as the general $k$-Pattern Domination problem, or being as easy as any pattern can be. 
This raises a question if we could provide a fine-grained dichotomy for this class of problems by showing that for each pattern $H$, the conditionally optimal running time for solving $H$-Dominating Set problem is either $mn^{k-2\pm o(1)}$, or $\big(\frac{m^{(k+1)/2}}{n}\big)^{1+o(1)}$.

As our last contribution, we answer this question negatively (assuming $k$-OV hypothesis), by tightly classifying the $k$-Induced Matching Domination problem that lies in neither of those two regimes, unless $k$-OV hypothesis fails.
More precisely, we show that this problem can be solved in running time $m^{\frac{k}{2}+o(1)}$ for all sufficiently large $k$, and provide a simple matching conditional lower bound by addapting the reduction from \autoref{prop:h-dom-lb}.

\section{Preliminaries}
\label{sec:prelims}
Let $n$ be a positive integer. 
We denote by $[n]$ the set $\{1,\dots, n\}$. 
If $S$ is an $n$-element set and $0\leq k\leq n$ is an integer, then $\binom{S}{k}$ denotes the set of all $k$-element subsets of $S$.


Let $\omega<2.3716$ \cite{fastmatrixmultiplication} denote the optimal exponent of multiplying two $n\times n$ matrices and $\MM(a,b,c)$ the time required to multiply two rectangular matrices of dimensions $a\times b$ and $b\times c$.
Note that if $\omega=2$, $\MM(a,b,c) \leq (ab + ac + bc)^{1+o(1)}$.
Let $\mathbb Z_{\leq d}[X]$ denote the set of all polynomials of degree at most $d$ whose coefficients are integers.
For a polynomial $f\in \mathbb Z_{\leq d}[X]$, the \emph{(maximum) degree} of $f$ is the largest exponent $r$ such that the term $X^r$ has a non-zero coefficient in $f$.
Symmetrically, the \emph{minimum degree} of $f$ denotes the smallest exponent $r$ such that the term $X^r$ has a non-zero coefficient in $f$.

For a graph $G$ and a vertex $v\in V(G)$, the \emph{neighbourhood} of $v$ is the set of vertices adjacent to $v$, denoted $N(v)$. 
The \emph{closed neighbourhood} of $v$, denoted $N[v]$ is defined as $N[v]:=N(v)\cup\{v\}$.
For the subset $S\subseteq V(G)$, we denote $N(S) := \bigcup_{v\in S} N(v)$ (resp. $N[S]:=\bigcup_{v\in S} N[v]$).
The \emph{degree} of $v$ denotes the size of its neighbourhood ($\deg(v) = |N(v)|$).
We further denote by $\deg^*(v)$ the size of the closed neighbourhood of the vertex $v$ ($\deg^*(v) = |N[v]|$). 
For any two vertices $u,v\in V(G)$, we denote by $d_G(u,v)$ the length of the shortest path between $u$ and $v$ in $G$.
The \emph{clique} (resp. \emph{independent set}) in a graph $G$ is a set of pairwise adjacent (resp. nonadjacent) vertices.
The \emph{Unbalanced $k$-Clique} problem, denoted \emph{$k$-Clique($\alpha_1,\dots, \alpha_k$)} is given a $k$-partite graph with part $i$ consisting of $\alpha_i$ vertices to decide if $G$ has a clique of size $k$. 

\section{$r$-Multiple $k$-Dominating Set} \label{sec:multiple-kds}
In this section, we provide the algorithms for the $r$-Multiple $k$-Dominating Set in sparse graphs. In particular, we prove a refined version of Theorem \ref{th:rleqk-2:algorithm-simple} and prove Theorem \ref{theorem:reduction-to-unbalanced-clique}.
We also show that the first algorithm cannot be significantly improved without violating some standard fine-grained hypotheses, by proving Theorem \ref{th:rleqk-2:lb} and finally show a conditional lower bound for the second algorithm.
\subsection{Algorithms}\label{section:mutliple-kds-algoritms}
All of our algorithms leverage the following simple lemma.
\begin{restatable}{lemma}{lemmaheavynodes}\label{lemma:heavy-nodes}
        For any fixed $k\geq 2$ and $r\le k$, any $r$-Multiple Dominating Set of size $k$ contains at least $r$ vertices $v_1,\dots, v_r$ with $\deg^*(v_i) \geq \frac{n}{k}$.
\end{restatable}
The proof of the lemma can be found in Appendix \ref{appendix-mult-dom}.
We call a vertex $v$ satisfying $\deg^*(v)\geq \frac{n}{k}$ a \emph{heavy vertex} and we let $\mathcal{H}$ denote the set of all heavy vertices.
A simple counting argument shows that there are at most $\bigO(\frac{m}{n})$ heavy vertices.

We distinguish between two cases based on the dependence of $r$ and $k$, namely $r\leq k-2$ and $r=k-1$
(note that if $r=k$, the problem becomes trivial), and in both cases we are able to show polynomial improvements in the sparse graphs.
Let us first consider the case $r\leq k-2$.
\subparagraph{Case $r\leq k-2$.}
To construct the desired algorithm, we modify the approach of~\cite{eisenbrand2004complexity,fischer2024effect} by employing polynomials to not only determine if a vertex is dominated by a set $D$, but also count how many vertices from $D$ are in its closed neighbourhood. We obtain the following refined version of Theorem~\ref{th:rleqk-2:algorithm-simple}.

\begin{restatable}{theorem}{ThRSmallAlgo}\label{th:rleqk-2:algorithm}
    For any fixed $k\geq 3$ and $r\leq k-2$, we can solve the $r$-Multiple $k$-Dominating Set in time
    \[
    \MM\Big(n^{\lceil \frac{k-r}{2}\rceil}\cdot\big(\tfrac{m}{n}\big)^{\lfloor\frac{r}{2}\rfloor}, n, n^{\lfloor \frac{k-r}{2}\rfloor}\cdot\big(\tfrac{m}{n}\big)^{\lceil\frac{r}{2}\rceil}\Big).
    \]
    If $\omega=2$, or if $k$ is sufficiently large, this running time becomes $(m/n)^r n^{k-r+o(1)}$. 
\end{restatable}
\begin{proof}
    Let $\mathcal{S}$ be the set consisting of all subsets of $V$ of size $\lceil \frac{k-r}{2}\rceil + \lfloor \frac{r}{2} \rfloor$ that contain at least $\lfloor \frac{r}{2} \rfloor$ heavy vertices and $\mathcal{T}$ be the set consisting of all subsets of $V$ of size $\lfloor \frac{k-r}{2}\rfloor + \lceil \frac{r}{2}\rceil$ that contain at least $\lceil \frac{r}{2} \rceil$ heavy vertices.
    By Lemma \ref{lemma:heavy-nodes}, any potential $r$-Multiple Dominating Set of size at most $k$ in $G$ can be written as a union of two elements $S\in \mathcal{S}, T\in \mathcal{T}$. Moreover, as argued above, we can bound the size of $\mathcal{S}$ and $\mathcal{T}$ as $|\mathcal{S}|\leq \bigO(n^{\lceil \frac{k-r}{2}\rceil}\cdot\big(\tfrac{m}{n}\big)^{\lfloor\frac{r}{2}\rfloor})$ and $|\mathcal{T}|\leq \bigO( n^{\lfloor \frac{k-r}{2}\rfloor}\cdot\big(\tfrac{m}{n}\big)^{\lceil\frac{r}{2}\rceil})$.
    We now construct the matrices $A$ and $B$ as follows.
    Let the rows of $A$ be indexed by $\mathcal{S}$ and columns by $V$ and set the entry $A[S,v]$ to $x^c\in \mathbb Z_{\leq k}[X]$ if and only if there are exactly $c$ elements in $S$ that are in the closed neighbourhood of $v$.
    Similarly let $B$ have columns indexed by $\mathcal{T}$ and rows by $V$ and set the entry $B[v,T]$ to $x^c\in \mathbb Z[X]$ if and only if there are exactly $c$ elements in $T$ that are in the closed neighbourhood of $v$.
    
    Let $C:=A\cdot B$.
    We observe that if $S,T$ are disjoint, then the coefficient of $x^c$ in $C[S,T]$ counts the number of vertices in $V$ that are dominated by exactly $c$ vertices from $S\cup T$.
    Hence, it suffices to verify if there exists a pair $S\in \mathcal{S}, T\in \mathcal{T}$ that are disjoint, such that the minimum degree of the polynomial $C[S,T]$ is $\geq r$.
    Moreover, it is straightforward to see that the degree of any entry in $C$ is bounded by $k=\bigO(1)$, and hence we can compute $C$ in the desired running time by applying the fastest matrix multiplication algorithm over the ring $\mathbb Z_{\leq k}[X]$. 
    The claimed running time follows.
\end{proof}
\subparagraph{Case $r=k-1$.}
By running the same algorithm as above, we can achieve a running time of $\bigO\big((\frac{m}{n})^{k-1+o(1)}n + n^2\big)$ (assuming $\omega = 2$, or sufficiently large $k$). 
%
However, perhaps surprisingly we can beat this running time significantly for larger $k$.
In fact, we proceed to show that for each $k\geq 3$, we can reduce the $(k-1)$-Multiple $k$-Dominating Set problem to an instance of $k$-Clique$(\frac{m}{n},\dots,\frac{m}{n},n)$.
To achieve that, we leverage the following simple observation.
\begin{observation}\label{obs:reduction-to-k-clique}
    For any fixed $k\geq 2$, let $x_1,\dots,x_k$ be any $(k-1)$-Multiple $k$-Dominating Set.
    Then for each $i\neq j$, vertices $x_i,x_j$ form a dominating set.
\end{observation}
Given a graph $G$, we can exploit this observation to preprocess the graph as follows.
Recall that $\mathcal{H}$ denotes the set of heavy vertices in our graph and by \autoref{lemma:heavy-nodes}, any $(k-1)$ multiple dominating set of size $k$ contains at least $(k-1)$ heavy vertices.
Let $V_1,\dots, V_{k-1}$ be copies of $\mathcal{H}$ and $V_k$ a copy of $V(G)$. 
Let $G' = (V',E')$, where $V' = V_1\cup\dots \cup V_k$ and for any pair $v_i\in V_i$, $v_j\in V_j$ (for $i\neq j$), add an edge between $v_i,v_j$ if and only if they form a dominating set in~$G$.
\begin{lemma}
    Let $G'$ be constructed as above.
    Then vertices $v_1,\dots, v_k$ form a clique in $G'$ if and only if they form a $(k-1)$-Multiple $k$-Dominating Set in $G$. 
\end{lemma}
\begin{proof}
    Assume first that some vertices $v_1,\dots, v_k$ form a clique in $G'$. We will call these vertices \emph{solution vertices}. 
    Take any vertex $w\in V(G)$ and assume that it is dominated by at most $k-2$ solution vertices.
    In particular, this means $w$ is not dominated by some pair of solution vertices $v_i, v_j$.
    However, this means that $v_i,v_j$ is not a dominating set and consequently, there is no edge between $v_i$ and $v_j$ in $G'$, contradicting that the solution vertices form a clique.
    The converse follows directly from \autoref{obs:reduction-to-k-clique}.
\end{proof}
By using the approach from \cite{fischer2024effect}, we can list all dominating sets of size $2$ in time $m^{\omega/2 + o(1)}$.
\begin{lemma} \cite{fischer2024effect}\label{lemma:2ds-listing}
    Given a graph $G$ with $n$ vertices and $m$ edges, there exists a randomized algorithm listing all dominating sets of size $2$ in time $m^{\omega/2+o(1)}$.
\end{lemma}
This gives us all the tools necessary to prove \autoref{theorem:reduction-to-unbalanced-clique}.
\TheoremReductionToUnbalancedClique*
\begin{proof}
    Note that it is sufficient to show that we can construct the graph $G'$ as defined above in time $m^{\frac{\omega}{2}+o(1)}$.
    Given a graph $G$, let $V_i,V_j$ be two arbitrary parts of $G'$ as described above. 
    Using the algorithm from \autoref{lemma:2ds-listing}, we can construct all the edges between the two parts in time at most $m^{\frac{\omega}{2} + o(1)}$ with high probability. 
    We repeat this for each pair $1\leq i<j\leq k$.
\end{proof}
Interestingly, this procedure also gives a polynomial improvement over brute-force in dense graphs.
\begin{corollary}
    We can solve $(k-1)$-Multiple $k$-Dominating Set in time $\bigO(n^{\omega \frac{k}{3}+1})$.
\end{corollary}
\autoref{theorem:reduction-to-unbalanced-clique} gives us a useful way to think about our problem in terms of Vertex Unbalanced $k$-Clique problem.
However, the question arises how to optimally solve this variation of $k$-Clique.
We partially answer this question by providing infinitely many values of $0\leq \gamma\leq 1$, such that if $\frac{m}{n} = \bigO(n^{\gamma})$, then for infinitely many values of $k$, we can solve this problem in time $\big((\frac{m}{n})^{k-1}n\big)^{\frac{\omega}{3}+o(1)}$, which is optimal under $k$-Clique Hypothesis (see \autoref{section:lower-bound} for details).
The idea is to apply the standard technique of grouping the vertices that form smaller cliques into three groups $W_1, W_2, W_3$ of roughly the same size, in such a way that there is a triangle between any three vertices $w_1\in W_1, w_2\in W_2, w_3\in W_3$ if and only if there are vertices $v_1\in V_1,\dots, v_k\in V_k$ that form a $k$-clique.
In order to be able to achieve this tightly, the values $k$ and $\gamma$ need to satisfy certain conditions.
\begin{restatable}{proposition}{PropositionUnbalancedCliqueConditions}
    Let $G=(V_1,\dots, V_k, E)$ be a $k$-partite graph with $|V_1| = n$ and $|V_2|=\dots = |V_{k}| = \bigO(n^\gamma)$ for some $0\leq \gamma\leq 1$. 
    If 
    \begin{enumerate}
        \item $(k-1 + \frac{1}{\gamma})$ is an integer divisible by $3$,\label{item:divisibility-by-3}
        \item $\frac{2}{\gamma}<k-1$, \label{item:bounding-gamma-k}
    \end{enumerate}
    then we can decide if $G$ has a clique of size $k$ in time $(n^{\gamma(k-1)+1})^{\frac{\omega}{3}+o(1)}$.
\end{restatable}
Notice that for each positive integer $p$, by setting $\gamma = \frac{1}{p}$, 
the first condition is satisfied by every $k \equiv 2p + 1 \pmod 3$, and the second condition is satisfied by every sufficiently large $k$ (in particular, for each of the infinitely many such choices of $\gamma$, both conditions can be satisfied by any of the infinitely many choices of $k$). 
For a detailed proof, see Appendix \ref{appendix-mult-dom}.
\subsection{Lower Bound} \label{section:lower-bound}
In this section, we show that the algorithms provided in the previous section are conditionally optimal. 
To do so, we introduce an intermediate problem, \emph{$r$-Multiple $k$-Orthogonal Vectors} defined as follows.
\begin{definition}[$r$-Multiple $k$-Orthogonal Vectors]
    Given sets $A_1,\dots, A_k \subseteq \{0,1\}^d$, determine if there exist vectors $a_1\in A_1, \dots, a_k \in A_k$ such that for each $t\in [d]$ there are pairwise distinct indices $i_1, \dots, i_r\in [k]$ with ${a_{i_1}[t] = \dots = a_{i_r}[t]=0}$.
\end{definition}
We note that when $r=1$, this problem is exactly the $k$-Orthogonal Vectors problem.
We can now adapt the reduction from \cite{fischer2024effect} to show that this problem reduces to a sparse instance of $r$-Multiple $k$-Dominating Set.
We note that we are using the \emph{moderate dimensional} variant of $r$-Multiple $k$-OV problem (i.e. $d = n^\delta$ for some small $\delta>0$).
\begin{lemma}
    For any fixed $k\geq 2$, $1\leq r\leq k-1$, let $A_1,\dots, A_k$ be a given instance of $r$-Multiple $k$-Orthogonal Vectors with $|A_1| = \dots = |A_r| = \bigO(\frac{m}{n})$ (for any $n\leq m\leq n^2$) and $|A_{r+1}| = \dots = |A_k| = n$. We can construct in linear time an equivalent instance of $r$-Multiple $k$-Dominating Set $G$ with $\bigO(n)$ vertices and $\bigO(m + dn)$ edges.
\end{lemma}
\begin{proof}
    Given an instance $A_1,\dots, A_k$ of $r$-Multiple $k$-Orthogonal Vectors, let $V(G) = X_1\cup \dots \cup X_k \cup D\cup R$ where the set $X_i$ corresponds to the set $A_i$, $D := [d]$ corresponds to the set of dimensions and $R$ is a set containing $(k+1)\binom{k}{r}$ vertices, representing ``redundant'' vertices.
    For each vertex $x_i\in X_i$ add an edge between $x_i$ and $t\in D$ if and only if the corresponding vector $a_i$ satisfies $a_i[t] = 0$.
    Partition $R$ into $\binom{k}{r}$ many sets $R_Q$ (for each $Q\in \binom{[k]}{r}$)
    of size $(k+1)$, and add an edge between any vertex $x_i\in X_i$ and any vertex $y\in R_Q$ if and only if $i\in Q$. It is straightforward to verify now that if $G$ has an $r$-multiple $k$-dominating set $S$, it must satisfy $|S\cap X_i| = 1$ for each $i\in [k]$.
    Finally for each vertex $x_i\in X_i$ for $i\leq r$ and $x_j\in X_j$ for $j\neq i$ add an edge between $x_i$ and $x_j$.

    It is straightforward to verify that the vectors $a_1\in A_1, \dots, a_k\in A_k$ satisfy the $r$-Multiple $k$-OV condition if and only if the corresponding vertices $x_1\in X_1,\dots, x_k\in X_k$ form an $r$-multiple dominating set in $G$.
    It remains to show that the constructed graph has $\bigO(n)$ many vertices and $\bigO(m + dn)$ many edges.
    Clearly, $G$ consists of $\bigO(r\frac{m}{n} + kn + d) = \bigO(\frac{m}{n} + n) = \bigO(n)$ many vertices and there are at most $\bigO(d (r\frac{m}{n} + kn) + r\frac{m}{n} \cdot kn) = \bigO(dn + m)$ edges.
\end{proof}
\begin{corollary}\label{cor:r-mult-kov-kds-reduction}
    Let $k\geq 2$, $1\geq r\leq k-1$ be fixed and $m = \Theta(n^{1+\gamma})$ for some $0< \gamma \leq 1$.
    If we can solve $r$-Multiple $k$-Dominating Set on graphs with $n$ vertices and $m$ edges in time $T(m,n)$, then 
    there exists a $\delta>0$, such that we can solve any instance $A_1,\dots, A_k$ of $r$-Multiple $k$-Orthogonal Vectors with $|A_1| = \dots = |A_r| = \bigO(\frac{m}{n})$, $A_{r+1} = \dots = A_k = n$ and $d=n^\delta$ in time $\bigO(T(m,n))$.
\end{corollary}
\begin{proof}
    Let $\delta = \gamma/2$ and given an instance $A_1,\dots, A_k$ of $r$-Multiple $k$-Orthogonal Vectors with $|A_1| = \dots = |A_r| = \bigO(\frac{m}{n})$, $A_{r+1} = \dots = A_k = n$ and $d=n^\delta$, apply the reduction from the previous lemma to obtain a graph with $\bigO(n)$ many vertices and $\bigO(m)$ many edges and run the algorithm solving $r$-Multiple $k$-Dominating Set in time $T(m,n)$ to this graph to obtain a $\bigO(T(m,n))$ algorithm for $r$-Multiple $k$-Orthogonal Vectors.
\end{proof}
It now remains to show that $r$-Multiple $k$-Orthogonal Vectors problem is conditionally hard.
In order to do this, we leverage the fine-grained classification of the first-order properties provided in \cite{bringmann2019fine} (for details see Appendix \ref{appendix-mult-dom}).
This allows us to prove the following result.
\begin{restatable}{lemma}{LemmaMultKOVHardness}
    Let $X_1,\dots, X_k$ be an instance of $r$-Multiple $k$-OV for $1\leq r \leq k-2$.
    There is no algorithm solving $r$-Multiple $k$-OV for $r\leq k-2$ in time $O(\big(|X_1|\cdot\dots \cdot |X_k|\big)^{1-\varepsilon})$ for any $\varepsilon>0$, unless the $(k-r+1)$-Uniform Hyperclique Hypothesis fails. This holds even when restricted to $|X_i|=\Theta(n^{\gamma_i})$ for an arbitrary choice of $\gamma_1, \dots, \gamma_k\in (0,1]$.
\end{restatable}
Finally, by combining this lemma and \autoref{cor:r-mult-kov-kds-reduction}, we can now conclude that our first algorithm is conditionally optimal (up to subpolynomial factors and resolution of matrix multiplication exponent) under the $3$-Uniform Hyperclique Hypothesis.
\TheoremrSmallLB*
Moreover, we observe that this implies that in dense graphs ($m=\Theta(n^2)$), there is no algorithm solving $r$-Multiple $k$-Dominating Set polynomially faster than brute force as long as $r\leq k-2$, unless $3$-Uniform Hyperclique hypothesis fails.

Notably, however, combining $r$-Multiple $k$-OV with the tools from \cite{bringmann2019fine} fails to provide a tight lower bound for $(k-1)$-Multiple $k$-Dominating Set in sparse graphs (for dense graphs we do get a tight classification, as discussed in Appendix \ref{appendix-mult-dom}). 
Nevertheless, by a careful reduction from the Independent Set problem, we can obtain a desired conditional lower bound. We sketch the reduction here. For the full proof see Appendix \ref{app-multdom-LB}.
\begin{restatable}{theorem}{ThMultDomCliqueReduction}\label{th:multdom-clique-reduction}
    Let $0< \gamma< 1$ be a rational number. 
    Then solving $(k-1)$-Multiple $k$-Dominating Set on graphs with $N$ vertices and $M=N^{1+\gamma}$ edges in time $\bigO\Big({\big(N^{\gamma(k-1)+1}\big)}^{\omega/3 - \varepsilon}\Big)$ for any $\varepsilon>0$ would refute $k$-Clique Hypothesis.
\end{restatable}
\begin{proofsketch}
    Write $\gamma = p/q$ for coprime positive integers $p,q$ and let $k^* := 3(k-1)p + 3q$. We reduce from $k^*$-Independent Set Detection. Let $G = (X_1,\dots, X_{k^*}, E)$ be a $k^*$-partite graph with $n$ vertices in each part.
    For each $i\in [k-1]$, let $A_i$ be the set consisting of all independent sets of size $3p$ from $X_{(i-1)\cdot 3p + 1}, \dots, X_{i\cdot 3p}$ and $A_k$ be the set consisting of all independent sets of size $3q$ from $X_{k^*-3q+1},\dots, X_{k^*}$.
    For each $i\in [k]$, let $V_i$ consist of nodes corresponding to the elements in $A_i$.
    Let $F$ be the set corresponding to the edge set $E$ of~$G$.
    We now construct a graph $G'$ as follows. Let $V(G') = V_1\cup\dots \cup V_k\cup F\cup R$, where $R$ is a gadget of size $\bigO(1)$ that ensures that if $G'$ contains any $(k-1)$-multiple dominating set of size $k$, it contains exactly one node from each set $V_i$.
    We add the remaining edges as follows.
    For any pair of nodes $v_i\in V_i$, $v_j\in V_j$ for $i\neq j$, add an edge between them.
    Finally, add an edge between a node $f\in F$ and $v_i\in V_i$ if and only if none of the vertices contained in the corresponding independent set $a_i\in A_i$ are among the two endpoints of the edge corresponding to $f$.
    By setting $N:=n^{3q}$, we can verify that $G'$ contains $\bigO(N)$ nodes and $\bigO(N^{1+\gamma})$ edges and that it contains a $(k-1)$-multiple dominating set of size $k$ if and only if $G$ has an independent set of size $k^*$.
    Finally, if there was an algorithm solving $(k-1)$-Multiple $k$-Dominating Set in time $\bigO\Big({\big(N^{\gamma(k-1)+1}\big)}^{\omega/3 - \varepsilon}\Big)$, by running the reduction above, we could solve the $k^*$-Independent Set problem in time 
    \[
    \bigO\Big(\big(N^{\gamma(k-1)+1}\big)^{\omega/3 - \varepsilon}\Big) =  \bigO\Big(\big(n^{3p(k-1)+3q)}\big)^{\omega/3 - \varepsilon}\Big) = \bigO\big(n^{k^*\omega/3 - \varepsilon'}\big)
    \]
    refuting the $k$-Clique Hypothesis.
\end{proofsketch}

\section{Dominating Patterns in Sparse Graphs} 
In this section we consider the Pattern Domination problem.
In particular, we first provide a simple argument that shows that for every pattern $H$ consisting of $k$ vertices (for $k\geq 3$), we can solve this problem in $mn^{k-2+o(1)}$ running time. 
On the lower bound side, we observe that the literature implicitly proves existence of a pattern $H$ for which this running time is optimal under the $k$-OV Hypothesis, thus settling the complexity of the case when the pattern $H$ is a part of the input.
We then consider the problem of detecting an $H$-Dominating Set for a fixed $k$-vertex graph $H$.
To this end, we show that for any fixed pattern $H$ consisting of $k$ vertices, the existence of an $O(\frac{m^{(k-1)/2+1-\varepsilon}}{n})$-time algorithm for this problem would refute the $k$-OV Hypothesis. We then show that this general lower bound is matched by a corresponding algorithm for some patterns $H$. 
The fine-grained complexity thus depends heavily on the structure of the graph $H$ itself, and we focus our attention to some of the most important patterns.

\begin{restatable}{proposition}{PropkPatternDomAlg}
Let $k\ge 7$. The $k$-Pattern Domination on graphs with~$n$ vertices and $m$ edges can be solved in time $O(mn^{k-2 + o(1)})$ (if $\omega = 2$, we can achieve this running time for all $k\geq 3$).
\end{restatable}
For a proof, see Appendix \ref{appendix-patt-dom}. On the other hand, it has been implicitly proved in \cite{fischer2024effect} that if $H$ is isomorphic to a complete bipartite graph $K_{1,(k-1)}$ (i.e. star graph on $k$ vertices), then detecting $H$-Dominating Set in time $O(mn^{k-2 - \varepsilon})$ for any $\varepsilon>0$ would refute $k$-OV Hypothesis, and thus in the general case, the algorithm above is the best we can do, up to subpolynomial factors, unless $k$-OV Hypothesis fails. 
We summarise this result in the following.
\begin{proposition}\cite{fischer2024effect}
    Let $H$ be a star graph on $k$ vertices.  
    Then for no $\varepsilon>0$ is there an algorithm solving $H$-Domination in time $mn^{k-2 - \varepsilon}$, unless $k$-OV Hypothesis fails.
\end{proposition}
The previous two propositions settle the fine-grained complexity of $k$-Pattern Domination in sparse graphs, but leave open an interesting research direction.
Namely, are there fixed patterns $H$ for which we can beat this running time, and if so, by how much.
Towards answering this question, we first provide a conditional lower bound, showing that for no pattern $H$ can we do better than $\big(\frac{m^{1+(k-1)/2}}{n}\big)^{1-o(1)}$ under $k$-OV hypothesis.
\PropHDomLB*
We adapt the reduction by Fischer, Künnemann and Redzic~\cite{fischer2024effect} to force any dominating set of size $k$ to induce the graph $H$.
For the detailed proof see Appendix \ref{appendix-patt-dom}.
\subsection{Dominating $k$-Clique and $k$-Independent Set}
In this section we consider the two classic graph patterns for which this problem has been well-studied, namely, $k$-Clique and $k$-Independent Set. 
Particularly, we settle the fine-grained complexity of both Dominating $k$-Clique and Dominating $k$-Independent Set by providing algorithms that match the conditional lower bound from \autoref{prop:h-dom-lb}.
Let us focus on the $k$-Clique case first.
In order to obtain a faster algorithm in sparse graphs, we leverage the following observation (for a proof see Appendix \ref{appendix-patt-dom}).
\begin{restatable}{observation}{ObsNumberOfkCliquesInGraph}[folklore]\label{obs:k-clique-number}
    A graph with $n$ vertices and $m$ edges has at most $\bigO(m^{\frac{k}{2}})$ $k$-cliques.
\end{restatable}
{\renewcommand\footnote[1]{}\cliqueDomAlg*}
The idea of the proof is to combine \autoref{obs:k-clique-number} with the observation that each dominating $k$-clique contains a heavy vertex, in order to restrict the size of our solution space.
After restricting the solution space, we follow the similar lines of the matrix multiplication algorithm for $k$-Dominating Set from \cite{eisenbrand2004complexity,fischer2024effect}.
For the full proof, see Appendix \ref{appendix-patt-dom}, Section \ref{app-section:clique}.

The last theorem shows that considering the density of the dominating pattern can be beneficial in obtaining a significant speedup over the standard $k$-Dominating Set algorithm, by observing that there are fewer such dense patterns (e.g. $k$-cliques) in sparse graphs.
On the other extreme of the density spectrum lie the independent sets.
There are typically many $k$-independent sets in sparse graphs ($\Omega(n^k)$), so we cannot use \autoref{obs:k-clique-number} to obtain a faster algorithm for Dominating $k$-Independent Set problem.
To nevertheless obtain a fast algorithm, we take advantage of one simple observation.
Namely, if we know that some fixed vertex $v$ is contained in some dominating $k$-independent set, by removing $N[v]$ from $G$, we can recursively obtain an instance of Dominating $(k-1)$-Independent Set problem, since no solution vertices will appear in $N[v]$ and moreover, $N[v]$ is already dominated by $v$.
As a technical note, the crucial reason why this approach fails for instances of the usual $k$-Dominating Set problem (without the restriction that the solution vertices induce a $k$-Independent Set) lies in the distinction between monochromatic and bichromatic versions of Dominating Set. 
In particular, after fixing a solution vertex $v$ of a dominating set, we no longer have to dominate the vertices from $N[v]$, but some of them might still appear in our solution.
We thus obtain an instance of Bichromatic $(k-1)$-Dominating Set (essentially a graph formulation of $k$-Set Cover), and it is known that this problem is hard already in very sparse instances (see, e.g.,~\cite{fischer2024effect}).
\begin{lemma}\label{lemma:domindset-recursion}
    Let $A_k(G)$ be an algorithm that finds a dominating $k$-independent set.
    Given a graph $G$, any dominating $k$-independent set containing a fixed vertex $v$ can be found by running $A_{k-1}(G-N[v])$.
\end{lemma}
This gives rise to a simple recursive algorithm whose time complexity we can bound as follows. (For a formal proof, we refer to Section~\ref{app-section:indset}.)
\begin{restatable}{lemma}{LemmaIndSetDomTC}
    Let $T_k(n,m)$ denote the running time of an algorithm solving Dominating $k$-Independent Set problem and $\mathcal{H}$ denote the set of heavy vertices.
    Then for every $k\geq 3$, the following inequality holds: 
    \[T_k(n,m) \leq \sum_{v\in \mathcal{H}}T_{k-1}(n - |N[v]|,m).\]
\end{restatable}
As the base case of our algorithm, we take the $k=2$ case, which we can solve in randomized time $m^{\omega/2 + o(1)}$.
\begin{lemma} \label{lemma:2ids-alg}
    There exists a randomized algorithm solving Dominating $2$-Independent Set in time $m^{\omega/2 + o(1)}$.
\end{lemma}
\begin{proof}
    By \autoref{lemma:2ds-listing}, we can list all dominating sets of size $2$ in time $m^{\omega/2 + o(1)}$ and for each we can in $\bigO(1)$ time check if it forms an independent set.
\end{proof}
We can now give a full algorithm with the analysis by exploiting the previous lemmas.
\indsetdomalg*
\begin{proof}
    If $k=2$ we apply \autoref{lemma:2ids-alg} to solve the problem in $m^{\omega/2+o(1)}$.
    For larger $k$, for each heavy vertex $v$, we ask if $G-N[v]$ contains an independent set of size $(k-1)$ that dominates $G$.
    If for any choice of $v$ the recursive algorithm returns YES, we return YES and otherwise return NO.
    We only have to analyse the time complexity.
    We know by the previous lemma that $T_2(n,m) = m^{\omega/2 + o(1)}$.
    For larger values of $k$, we have
    \begin{align*}
        T_k(n,m) & \leq \sum_{v\in \mathcal{H}}T_{k-1}(n - |N[v]|,m)\\
                 & \leq \min(n,\frac{m}{n})\cdot \max_{v\in \mathcal{H}}T_{k-1}(n - |N[v]|,m)\\
                 &\leq \frac{m}{n}\cdot \max_{\delta \in [0,1]}T_{k-1}(n^\delta,m)
    \end{align*}
    Now it only remains to bound the value $\max_{\delta \in [0,1]}T_{k-1}(n^\delta,m)$.
    If $k\leq 3$, this value is bounded by $m^{\omega/2+o(1)}$ and we obtain the claimed running time. 
    So assume that $k\geq 4$ and consider two separate cases, namely when $n^\delta<\sqrt{m}$ and when $n^\delta\geq \sqrt{m}$.
    In the former case, we can simply list all dominating sets of size $k-1$ in time $n^{k-1 + o(1)}<m^{\frac{k-1}{2} + o(1)}$ (assuming $\omega = 2$) using the algorithm from \cite{PatrascuW10}\footnote{If $k\geq 8$, we can obtain this running time, even with the current value of $\omega$.}, and this again yields a running time of $(\frac{m^{(k-1 + \omega)/2}}{n})^{1+o(1)}$.
    In the latter case we can proceed inductively, since $\frac{m}{n^\delta}\leq \sqrt{m}$, we have that $T_{k-1}(n^\delta,m) \leq \sqrt{m}\max_{\delta' \in [0,1]}T_{k-2}(n^{\delta'},m)$ and we can bound $\max_{\delta' \in [0,1]}T_{k-2}(n^{\delta'},m)\leq m^{(k-2)/2}$ by a simple induction on $k$, yielding the desired running time.
\end{proof}

\subsection{Dominating $k$-Induced Matching}
So far we considered three different pattern classes (cliques, independent sets and stars), and in two out of those three cases we can obtain an algorithm that runs in $\big (\frac{m^{(k+1)/2}}{n}\big)^{1+o(1)}$ (if $\omega = 2$), which is the best we can do for any pattern assuming $k$-OV Hypothesis, and in the remaining case we can show an $mn^{k-2-o(1)}$ conditional lower bound, which makes this pattern as hard as any pattern can be.
This suggests that there might be a dichotomy of all $k$-vertex graphs into two classes:
\begin{enumerate}
    \item Easy Patterns (those for which there exists an algorithm solving $H$-Dominating Set in $\big (\frac{m^{(k+1)/2}}{n}\big)^{1+o(1)}$)
    \item Hard Patterns (those for which we can show an $mn^{k-2-o(1)}$ lower bound under $k$-OV Hypothesis).
\end{enumerate}
In this section we show that such a dichotomy is unlikely.
More precisely, we find a pattern which is in neither of those two categories (unless $k$-OV Hypothesis fails).
Let $k$-induced matching be a graph consisting of $k/2$ independent edges.
In this section we prove that we can solve the Dominating $k$-Induced Matching problem in $m^{k/2 + o(1)}$ running time (if $\omega=2$) and provide a matching conditional lower bound under $k$-OV Hypothesis.
Due to lack of space, we only state our results; the full proofs can be found in Appendix \ref{appendix-patt-dom}.
\begin{restatable}{theorem}{ThInducedMatchingAlgorithm}
    Given a graph $G$ with $n$ vertices and $m$ edges, we can solve Dominating $k$-Induced Matching in time 
    \[\
    \MM(m^{\lceil\frac{k}{2}\rceil}, n, m^{\lfloor \frac{k}{2}\rfloor}).
    \]
    If $\omega=2$, this running time becomes $m^{\frac{k}{2}+o(1)}$ for every even $k\geq 4$.
\end{restatable}
Finally, we show that this running time cannot be significantly improved, unless $k$-OV Hypothesis fails.
To achieve this, we apply a simple adaptation of the reduction from \autoref{prop:h-dom-lb}.
\begin{restatable}{theorem}{ThInducedMatchingLB}
    For no even $k\geq 4$ and $\varepsilon>0$ is there an algorithm solving Dominating $k$-Induced Matching in time $O(m^{\frac{k}{2}-\varepsilon})$, unless the $k$-OV hypothesis fails.
\end{restatable}

\bibliographystyle{plainurl}
\bibliography{refs}
\appendix
\section{Hardness Assumptions}\label{app-hardness}
Consider the $k$-Orthogonal Vectors problem ($k$-OV) that is stated as follows.
Given $k$ sets $A_1,\dots, A_{k}$ of $d$-dimensional binary vectors, decide whether there exist vectors $a_1\in A_1,\dots, a_k\in A_k$ such that for all $t\in [d]$, it holds that $\prod_{i=1}^{k} a_i[t] = 0$.
A simple brute force approach solves the $k$-OV in time $\bigO(d\cdot \prod_{i\in [k]}|A_i|)$.
On the other hand, it is known that for sufficiently large $d$ ($d=\log^2(|A_1| + \dots + |A_k|)$), any polynomial improvement over this running time would refute SETH.
\begin{conjecture}[$k$-OV Hypothesis]
    For no $\varepsilon>0$ and for no $0\leq \gamma_1,\dots,\gamma_k\leq 1$ is there an algorithm solving $k$-OV with $|A_1| = n^{\gamma_1}, \dots, |A_k| = n^{\gamma_k}$, $d=\log^2 n$ in time $\bigO(n^{(\sum_{i=1}^{k}\gamma_i)-\varepsilon})$.
\end{conjecture}
Typically, the $k$-OV hypothesis is stated for the special case for $\gamma_1 = \dots = \gamma_k = 1$, which we refer to as \emph{Balanced $k$-OV Hypothesis}. 
However, these two hypotheses are known to be equivalent \cite{BringmannFOCS15,fischer2024effect}, and for the purposes of this paper, we benefit from using the more general one.

The $k$-Clique Detection problem is given a graph $G$ with $n$ vertices to decide if $G$ contains a clique of size $k$. 
For the special case of $k=3$, there is a simple matrix-multiplication based algorithm that detects triangles in $n^\omega$ running time. 
Moreover, for larger $k$ (divisible by $3$), one can solve $k$-Clique Detection in $n^{\omega k/3}$ by reducing to Triangle Case \cite{nevsetvril1985complexity,eisenbrand2004complexity}.
Notably, no improvements over these simple algorithms have been made for decades, thus suggesting that they might be optimal and leading to the following hypothesis (see e.g. \cite{abboud2018if}).
\begin{conjecture}[$k$-Clique Hypothesis]
    For no $\varepsilon>0$ and $k\geq 3$ is there an algorithm solving $k$-Clique Detection in time $O(n^{k\omega/3 - \varepsilon})$.
\end{conjecture}
The $h$-Uniform $k$-Hyperclique Detection problem is given an $h$-uniform hypergraph $G$ with $n$ vertices to decide if $G$ contains a hyperclique of size $k$ (i.e. $k$ vertices $x_1,\dots, x_k$ such that each $h$-tuple $x_{i_1},\dots, x_{i_h}$ of pairwise distinct $i_1,\dots, i_h\in [k]$ forms an edge in $G$).
When dealing with $h$-uniform hypergraphs, it turns out that the same matrix multiplication techniques used for $k$-Clique Detection cannot be used to improve over brute-force.
In fact, for $h\geq 3$, no algorithm running in $n^{k-\varepsilon}$ is known to be able to detect if an $n$ vertex graph $G$ contains an $h$-uniform hyperclique of size $k$, and moreover any such improvement would cause a breakthrough in other notoriously hard problems as well, most notably Max-$h$-SAT and Max Weight $k$-Clique (see e.g. \cite{lincoln2018tight} for a dedicated discussion on hardness of hyperclique detection).
This prompts the following hypothesis.
\begin{conjecture}[$h$-Uniform $k$-Hyperclique Hypothesis]
    For no $\varepsilon>0, h\geq 3, k\geq h+1$ is there an algorithm solving $h$-Uniform $k$-Hyperclique Detection in time $O(n^{k - \varepsilon})$.
\end{conjecture}
For the purposes of this paper, we need a seemingly slightly more general assumption. Specifically, we assume that we cannot detect an $h$-uniform $k$-hyperclique in a $k$-partite graph $G = (V_1,\dots, V_k, E)$ 
significantly faster than brute-force.
\begin{conjecture}[Unbalanced $h$-Uniform $k$-Hyperclique Hypothesis]
    For no $\varepsilon>0, h\geq 3, k\geq h+1$ is there an algorithm solving $h$-Uniform $k$-Hyperclique Detection in $k$-partite graph $G = (V_1,\dots, V_k, E)$ in time $O(\big(\prod_{i\in [k]}|V_i|\big)^{1-\varepsilon})$.
\end{conjecture}
However, it turns out that these two assumptions are equivalent in a sense that refuting the Unbalanced $h$-Uniform $k$-Hyperclique Hypothesis would refute the $h$-Uniform $k$-Hyperclique Hypothesis and vice-versa.
The proof is a straightforward self-reduction and is analogous to the proof that the $k$-OV Hypothesis is equivalent to the Unbalanced $k$-OV Hypothesis, see~\cite{BringmannFOCS15,fischer2024effect}.
\section{$r$-Multiple $k$-Dominating Set}\label{appendix-mult-dom}
\subsection{Algorithms}
\lemmaheavynodes*
\begin{proof}
    Let $v_1,\dots, v_k$ be any $r$-Multiple Dominating Set of size $k$, and assume $\deg(v_i)\geq \deg(v_j)$ for any $i\leq j$. 
    From the definition of $r$-Multiple $k$-Dominating Set, we can obtain the following inequality:
    \begin{align} \label{eq:php}
        \sum_{i=1}^{k} \deg^*(v_i) \geq rn.
    \end{align}
    Assume for contradiction that $\deg^*(v_i)\leq \frac{n}{k} - 1$ for each $i\geq r$. 
    We can construct the following chain of inequalities.
    \begin{align*}
        \sum_{i=1}^{k} \deg^*(v_i) &\leq (r-1)\cdot n + (\frac{n}{k}-1)(k-r+1)& \\
        &= (r-1)\cdot n + n - \frac{rn}{k} + \frac{n}{k} - k + r -1&  \\
        &= r\cdot n - \frac{n}{k}(r-1) - (k-r + 1)& (r\geq 1) \\
        &\leq r\cdot n - (k-r+1) & (r\leq k) \\
        & \leq rn-1 &
    \end{align*}
    Contradicting \ref{eq:php} and concluding the proof.
\end{proof}
\PropositionUnbalancedCliqueConditions*
\begin{proof}
    Let $\alpha$ and $\beta$ be positive integers satisfying $\alpha + 2\beta + 1= k$ (we will set the exact values shortly).
    Let $W_1$ be the set consisting of all cliques $v_1\in V_1,\dots, v_{\alpha+1}\in V_{\alpha+1}$ of size $\alpha+1$.
    Similarly, let $W_2$ and $W_3$ be the sets consisting of cliques of size $\beta$ from sets $V_{\alpha + 2}, \dots, V_{\alpha + \beta + 1}$ and $V_{\alpha + \beta + 2}, \dots, V_{k}$ respectively.
    Consider the tripartite graph $G'$ consisting of sets $W_1,W_2,W_3$ by adding an edge between any pair $w_i\in W_i, w_j\in W_j$ for $i\neq j$ if and only if $w_i \cup w_j$ form a clique in $G$.
    It is now straightforward to verify that $G'$ contains a triangle if and only if $G$ contains a clique of size $k$.
    The only thing that now remains is to determine the values $\alpha$ and $\beta$ such that $W_1,W_2,W_3$ are all of the same size (asymptotically).
    \begin{claim}
        If $\beta = \frac{(k - 1 + \frac{1}{\gamma})}{3}$ and $\alpha = k-1-2\beta$, then $|W_1| = |W_2| = |W_3| = n^{\gamma\beta}$.
    \end{claim}
    \begin{proof}
        We first observe that by condition \ref{item:divisibility-by-3}, both $\alpha$ and $\beta$ are integers and moreover by condition \ref{item:bounding-gamma-k}, they are both positive.
        Notice that now it is sufficient to prove that $n\cdot n^{\gamma\alpha} = n^{\gamma \beta}$ (equivalently $1+\gamma\alpha = \gamma\beta$).
        This can easily be verified as:
        \begin{align*}
            \gamma \beta &= (\gamma(k-1) + 1)/3 \\
            &=\gamma(k-1) + 1 - \frac{2}{3}(\gamma(k-1) + 1) \\
            &= \gamma(k-1) + 1 - 2 \gamma \beta\\
            &= 1 + \gamma (k-1-2\beta) \\
            &= 1 + \gamma \alpha.
        \end{align*}
    \end{proof}
    A standard matrix multiplication algorithm can now detect triangles in $G'$ in time 
    \[(n^{\gamma\beta})^{\omega+o(1)} = \big(n^{\gamma(k-1)+1}\big)^{\omega/3 + o(1)}\]
    as desired.
\end{proof}
\subsection{Lower Bounds} \label{app-multdom-LB}
We start this section by proving the conditional hardness of $r$-Multiple $k$-Orthogonal Vectors problem. 
As we already remarked, for $r=1$, it is equivalent to a well-known problem, namely, $k$-Orthogonal Vectors, for which any algorithm running in time $\big(|A_1|\cdot\dots\cdot|A_k|\big)^{1-\varepsilon}$ would refute SETH, a long standing hardness conjecture.

The goal now is to show some similar conditional hardness result of this problem for general $r$.
In order to achieve this, we use the concept of $h$-hardness introduced in \cite{bringmann2019fine}. 
For a propositional formula $f(z_1,\dots, z_k)$ and an index set $I\subseteq [k]$, an \emph{$I$-restriction} of $f$ is a formula obtained from $f$ after substituting the variables $z_i$ for every $i\in I$ by constant values from $\{0,1\}$.
\begin{definition}
    A propositional formula $f(z_1,\dots, z_k)$ is $h$-hard, $0 \leq h \leq k$, if for any index set $I\in {[k]\choose k-h}$,  there exists some $I$-restriction of $f$ with exactly one falsifying assignment.
\end{definition}
Suppose we are given a $(k+1)$-partite graph $G = (X_1,\dots, X_k,Y, E)$. Define 
\[
\psi := (\exists x_1\in X_1)\dots(\exists x_k\in X_k)(\forall y\in Y)\phi(x_1,\dots, x_k,y),
\]
where $\phi$ is an arbitrary Boolean formula defined over the \emph{edge relations} $E(v,y)$ for $v\in \{x_1,\dots, x_k\}$. 
Note that we can view $\phi$ as a Boolean function $\{0,1\}^k\to \{0,1\}$ that maps the values of $E(x_i,y)$ to a truth value and hence it makes sense to talk about $h$-hardness of $\phi$.

In \cite{bringmann2019fine}, Bringmann, Fischer and Künnemann provide a full fine-grained classification of such properties $\psi$, based on the maximum value $h$ for which $\phi$ is $h$-hard.
Note that for our use-case we don't need the full classification, but the following lemma that follows directly from their classification suffices.
\begin{lemma} \label{lemma:h-hardness}
    Given a graph $G$ and a property $\psi$ as above, let $h$ be the maximum value such that $\phi$ is $h$-hard. 
    Assume we can decide if there are vertices $x_1\in X_1, \dots, x_k\in X_k$ that satisfy $\phi(x_1,\dots, x_k, y)$ for every $y$ in time $T(|X_1|, \dots, |X_k|, |Y|)$. Then
    \begin{enumerate}
        \item \label{item:hyperclique}If $3\leq h$, 
        \[
            T(|X_1|, \dots, |X_k|, |Y|) \geq (|X_1|\cdot \dots \cdot |X_k|)^{1-o(1)},
        \]
        unless the $h$-Uniform Hyperclique hypothesis fails.
        \item \label{item:clique}If $2\leq h< k$, and $|X_1| = \dots =|X_k| = n$, then
        \[
            T(|X_1|, \dots, |X_k|, |Y|) \geq n^{\omega\frac{k}{3}-o(1)}
        \]
        unless the $k$-Clique hypothesis fails.
    \end{enumerate}
\end{lemma}
We can observe that $r$-Multiple $k$-OV can be stated equivalently as follows:
\[
    (\exists x_1\in X_1)\dots(\exists x_k\in X_k)(\forall y\in Y)\big(\bigvee_{1\leq i_1<\dots<i_r\leq k} E(x_{i_1},y)\land \dots \land E(x_{i_r},y)\big)
\]
It is now sufficient to argue $h$-hardness of the propositional formula in the statement of $r$-Multiple $k$-OV above, for every dependence of $r$ and $k$.
\begin{lemma}
    Let \[
    \phi:= \bigvee_{1\leq i_1<\dots<i_r\leq k} E(x_{i_1},y)\land \dots \land E(x_{i_r},y).\]
    For every $r\leq k-1$, $\phi$ is $(k-r+1)$-hard
\end{lemma}
\begin{proof}
    Let $I\in {[k]\choose r-1}$ be arbitrary. 
    Set $E(x_i,y)$ to $1$ for every $i\in I$.
    It is now straightforward to verify that setting any $x_j$ to $1$ for $j\not\in I$ satisfies $\phi$, hence the all-$0$ assignment is the unique falsifying assignment as desired.
\end{proof}
By combining the last two lemmas, we obtain the following.
\LemmaMultKOVHardness*
Furthermore, based on \autoref{lemma:h-hardness}, \autoref{item:clique}, in combination with \autoref{cor:r-mult-kov-kds-reduction} we can also conclude that under $k$-Clique hypothesis, our algorithm for $(k-1)$-Multiple $k$-Dominating Set is optimal in the \emph{dense} setting for every $k\geq 3$.
\begin{proposition}
    For no $\varepsilon>0$ and $k\geq 3$ is there an algorithm solving $(k-1)$-Multiple $k$-Dominating Set in time $O(n^{\omega\frac{k}{3}-\varepsilon})$, unless $k$-Clique Hypothesis fails.
\end{proposition}
\ThMultDomCliqueReduction*
\begin{proof}
    Write $\gamma$ as $p/q$ for coprime positive integers $p,q$ (since $\gamma$ is rational, there is a unique way to do this).
    Let $k' : = (k-1)p + q$ and $d$ be a positive integer that will be fixed later.
    We reduce from $(dk')$-Independent Set Detection.
    Let $G = (X_1,\dots, X_{dk'}, E)$ be a $(dk')$-partite graph with $n$ vertices in each part.
    We now proceed to construct a graph $G'$ that will represent the equivalent instance of $(k-1)$-Multiple $k$-Dominating Set.
    First partition the set $[dk']$ into $(k-1)$ sets $P_1, \dots, P_{k-1}$ of size $d\cdot p$ and one set $P_k$ of size $d\cdot q$.
    Now construct the sets $V_1,\dots, V_k$ as follows.
    For $i\leq k-1$, each set $V_i$ contains all independent sets of size $d\cdot p$ of the form $x_{i_1}\in X_{i_1}, \dots, x_{i_{dp}}\in X_{i_{dp}}$ such that $P_i = \{i_1,\dots, i_{dp}\}$.
    Finally, $V_k$ consists of all independent sets of size $d\cdot q$ satisfying the similar constraints.
    Let $F$ be the set containing a vertex $e$ corresponding to each edge from $E(G)$.
    Let $R$ be a set consisting of $k$ copies of set $[k+1]$.
    We now set the vertex set of $G'$ as $V(G') = V_1\cup \dots \cup V_k \cup F \cup R$.
    To make a distinction between the vertices of $G$ and the vertices of $G'$, we will call the vertices of $G'$ \emph{nodes}.

    We can now construct the edges of $G'$.
    Partition the set $R$ into $k$ sets $R_1,\dots, R_k$ of size $(k+1)$ and add an edge between each pair of nodes $r_i\in R_i$ and $v_j\in V_j$ for $i\neq j$.
    Since nodes in $R$ will be incident to no other edges, they assure that if there exists a $(k-1)$-Multiple Dominating Set in $G'$, it contains exactly one node from each $V_i$ \footnote{Let $S$ be a $(k-1)$-Multiple Dominating Set. If $|V_i\cap S|\geq 2$, then nodes in $R_i$ are dominated by at most the remaining $k-2$ nodes from $S$, contradiction. If $V_i\cap S = \emptyset$, then each $R_j$ for $i\neq j$ contains at least one node that is dominated by at most $k-2$ nodes from $S$, using that each $V_r$ satisfies $|S\cap V_r| \leq 1$.}.
    Moreover, add an edge between any pair of nodes $v_i\in V_i$ and $v_j\in V_j$ for $i\neq j$.
    Finally, add an edge between a node $f\in F$ and $v_i\in V_i$ if and only if none of the vertices in the independent set corresponding to $v_i$ are among the two endpoints of the edge corresponding to $f$.
    Formally, let $f$ correspond to an edge $\{x,y\}\in E(G)$, then we can write $\{f,v_i\}\in E(G')$ if and only if $\{x,y\}\cap v_i = \emptyset$.
    \begin{claim}
        Let $N = n^{dq}$ and $M = N^{1+\gamma}$. 
        $G'$ has $\bigO(N)$ nodes and $\bigO(M)$ edges. 
    \end{claim}
    \begin{proof}
        Since $\gamma< 1$, we also have $p< q$. 
        In total, $G'$ has at most 
        \[{(k-1)n^{dp} + n^{dq} + n^2 + k(k+1) = \bigO(n^{dq}) = \bigO(N)}\]
        many nodes, and 
        \[\bigO(n^{dp}\cdot n^{dq} + n^{dq + 2}) = \bigO(n^{dp + dq}) = \bigO(N^{\gamma + 1})\]
        many edges, where $n^{dq + 2} \leq \bigO(n^{dp+dq})$ follows by setting $d\geq 2$.
    \end{proof}
    \begin{claim}
        There exists a $(k-1)$-multiple $k$-dominating set $S$ in $G'$ if and only if there exists an independent set $x_1\in X_1, \dots, x_{dk'}\in X_{dk'}$ of size $dk'$ in $G$.
    \end{claim}
    \begin{proof}
        Assume first that there exists an independent set $x_1\in X_1, \dots, x_{dk'}\in X_{dk'}$ of size $dk'$ in $G$.
        By construction, we can partition this set into $k$ subsets, such that each subset corresponds to a unique node $v_i\in V_i$. 
        We claim that such nodes $v_1,\dots, v_k$ form a $(k-1)$-multiple $k$-dominating set. 
        Indeed, clearly by construction, all nodes in $V_i$ and $R_i$ are dominated by each $v_j$ for $j\neq i$.
        It remains to show that the nodes in $F$ are also dominated by at least $(k-1)$ $v_i$'s. 
        Assume this is not the case.
        That is, assume that there exists a node $f$ corresponding to an edge $\{x,y\}$ in $G$ that is dominated by at most $(k-2)$ $v_i$'s. 
        It is straightforward to observe that no vertex from $G$ can appear in two nodes $v_i\in V_i, v_j\in V_j$ for $i\neq j$.
        Hence, by construction of edges between $F$ and $V_i$'s, this is only possible if $x\in v_i$ and $y\in v_j$ for distinct $i,j$.
        But this would imply that the subgraph of $G$ induced by $v_i\cup v_j$ contains an edge, contradicting the assumption that $x_1,\dots, x_{dk'}$ form an independent set.

        Conversely, assume that no tuple  $x_1\in X_1, \dots, x_{dk'}\in X_{dk'}$ of size $dk'$ forms an independent set in $G$ and assume for contradiction that $G'$ admits a $(k-1)$-multiple $k$-dominating set $S$.
        As noted above, any $(k-1)$-multiple $k$-dominating set satisfies $|S\cap V_i| = 1$ for each $V_i$.
        We can thus label the nodes in $S$ as $v_1,\dots, v_k$ where each $v_i$ comes from $V_i$. 
        Consider now the subgraph of $G$ of size $dk'$ induced by $v_1\cup\dots \cup v_k$. By assumption that $G$ contains no independent set of size $dk'$, this subgraph contains an edge $\{x,y\}$. 
        Consider now the node $f\in F$ that corresponds to this edge. 
        There are at least $(k-1)$ nodes in $S$ that are adjacent to $f$ and hence at most one node $v_i$ is non-adjacent to $f$.
        But by construction of $G'$ this means that the independent set that corresponds to $v_i$ contains both vertices $x$ and $y$, which yields a contradiction, since $x$ and $y$ are adjacent.
        \end{proof}
        Finally, assume that there is an algorithm solving $(k-1)$-Multiple $k$-Dominating Set in time $\big(N^{\gamma (k-1)+1}\big)^{\omega/3-\varepsilon}$ for some $\varepsilon>0$.
        Then given a graph $G$ as above, using the construction above, we could decide if $G$ has an independent set $x_1\in X_1,\dots, x_{dk'}\in X_{dk'}$ of size $dk'$ in time 
        \begin{align*}
            \bigO(\big(N^{\gamma (k-1)+1}\big)^{\omega/3-\varepsilon}) &\leq \big(n^{dq\gamma (k-1)+dq}\big)^{\omega/3-\varepsilon'} & \text{(for e.g. $\varepsilon'=2\varepsilon$.)} \\
            & = \big(n^{dp(k-1)+dq}\big)^{\omega/3-\varepsilon'} & (\gamma=p/q)\\
            & = n^{dk'\omega/3-\varepsilon''} & (k' = p(k-1) + q)
        \end{align*}
        Thus refuting the $k$-Clique Hypothesis.
\end{proof}
\section{Dominating Patterns in Sparse Graphs}\label{appendix-patt-dom}
\begin{lemma}\label{lemma:listing-ds-h-dom}
    Given a graph $G$ with $n$ vertices and $m$ edges, let $T_k(m,n)$ be the time complexity of listing all dominating sets of $G$ of size $k$ (for a fixed constant $k$).
    Then $k$-Pattern Domination on $G$ can be solved in time $\bigO(T_k(m,n))$.
\end{lemma}
\begin{proof}
    Follows by observing that given a $k$-vertex graph $H'$, it takes $f(k) = \bigO(1)$ time to check if it is isomorphic to $H$.
\end{proof}
\PropkPatternDomAlg*
\begin{proof}
    By \cite{fischer2024effect}, we can list all dominating sets of size $k$ in time $\MM(n^{\lceil\frac{k-1}{2}\rceil}, n, mn^{\lfloor \frac{k-1}{2} \rfloor})$. 
    The desired now follows from \autoref{lemma:listing-ds-h-dom}.
\end{proof}
\PropHDomLB*
\begin{proof}
    Let $A_1,\dots, A_{k-1}$ denote sets of $d$-dimensional (assume $d=\bigO(\log^2 n)$) binary vectors of size $\sqrt{m}$ and $A_k$ be a set of $d$-dimensional binary vectors of size $\frac{m}{n}$ for $n\leq m \leq n^2$.
    We construct a graph $G$ with $\Tilde{\bigO}(n)$ vertices and $\Tilde{\bigO}(m)$ edges that has $H$-dominating set if and only if we can find vectors $a_1\in A_1,\dots, a_k\in A_k$ such that $\prod_{i=1}^k a_i[j] = 0$ for every $j = 1,\dots, d$.

    Label the vertices of $H$ by $y_1,\dots, y_k$.
    Let $V(G) = X_1\cup \dots\cup X_k \cup D \cup R$, where $X_i$ is a copy of the set $A_i$, $D$ is a copy of the set $[d]$  and $R$ is a copy of the set $[n+k^2-1]$. 
    We will refer to the vertices from $X_i$ as \emph{vector vertices}, those from $D$ as \emph{dimension vertices} and those from $R$ as \emph{redundant vertices}.
    We now add the edges as follows.
    Partition the redundant vertices into $k$ parts $R_1,\dots, R_k$, such that $|R_1| = \dots = |R_{k-1}| = k+1$ and $|R_k| = n$.
    Add an edge between each pair $x_i\in X_i$ and $r_i\in R_i$. 
    It is straightforward to verify that these edges imply that if $G$ contains a dominating set $S$ of size $\leq k$, each set $X_i$ satisfies $|X_i\cap S| = 1$. 
    For any pair $x_i, x'_i\in X_i$, add an edge between them.
    Moreover, for any pair $x_i\in X_i$, $x_j\in X_j$ for $i\neq j$, add an edge between them if and only if $(y_i, y_j)\in E(H)$.
    Notice that this ensures that any selection of vector vertices $x_1\in X_1, \dots, x_k\in X_k$ induces $H$. 
    Finally, we add edges between the vector and the dimension vertices naturally, i.e. for each $x_i\in X_i$ and $t\in [d]$, add an edge between $x_i$ and $t$ if and only if the vector $a_i$ corresponding to $x_i$ satisfies $a_i[t] = 0$.
    We first prove that the graph $G$ has $\bigO(n)$ vertices and $\bigO(m)$ many edges. 
    Indeed, each of the sets $X_1\cup \dots\cup X_k \cup D \cup R$ has $\bigO(n)$ vertices and since $k\in \bigO(1)$, the bound follows.
    Moreover, since each edge in $G$ has an endpoint in some $X_i$, the number of edges is at most the sum of degrees of vertices in the vector sets. 
    It is easy to verify that this number is at most
    \begin{align*}
    (k-1)\sqrt{m}\big((k-1)\sqrt{m} + \frac{m}{n} + (k+1) + d\big) + \frac{m}{n}\big((k-1)\sqrt{m} + n + d\big) = \bigO(m).
    \end{align*}

    We finally prove the correctness.
    Assume first that there exist vectors $a_1\in A_1,\dots, a_k\in A_k$ such that $\prod_{i=1}^k a_i[j] = 0$ for every $j = 1,\dots, d$.
    Consider the corresponding vertices $x_1\in X_1,\dots, x_k\in X_k$.
    As argued above, vertices $x_1,\dots, x_k$ induce a subgraph of $G$ that is isomorphic to $H$, so it remains to verify that they dominate the whole graph.
    First note that each of the redundant vertices in $R_i$ is dominated by $X_i$. 
    Furthermore, each vertex in $X_i$ is dominated by $x_i$. 
    Consider any dimension vertex $t\in D$.
    Since $\prod_{i=1}^k a_i[t] = 0$, some vector $a_i$ satisfies $a_i[t] = 0$, and by construction the corresponding vector $x_i$ is adjacent to $t$.
    Hence, we may conclude that the vertices $x_1,\dots, x_k$ form an $H$-Dominating Set.
    Conversely, assume that there $G$ has an $H$-Dominating Set. As argued above, this set can only consist of vertices $x_1\in X_1, \dots, x_k\in X_k$. 
    We claim that the corresponding vectors $a_1\in A_1, \dots, a_k\in A_k$ satisfy $\prod_{i=1}^k a_i[t] = 0$ for each $t\in [d]$. 
    Indeed, if we fix any $t\in [d]$, the dimension vertex that corresponds to $t$ is adjacent to at least one $x_i$ and hence the corresponding vector satisfies $a_i[t] = 0$.
\end{proof}
\subsection{Dominating $k$-Clique} \label{app-section:clique}
\ObsNumberOfkCliquesInGraph*
\begin{proof}
    We prove this by strong induction on $k$. 
    For $k=2$, the bound is trivial, and moreover, for $k=3$ it is well-known (see e.g. \cite{itai1978trianglelisting}).
    Assume now that for some $k\geq 4$ this bound holds for the number of cliques of size $r$, for all values $2\leq r<k$.
    We show that it also holds for $k$-cliques.
    In particular, any $k$-clique $C$ can be obtained by concatenating a $\lceil \tfrac{k}{2}\rceil$-clique $C_1$ with a $\lfloor \tfrac{k}{2}\rfloor$-clique $C_2$.
    Note that since $k\geq 4$, we have $2\leq \lfloor \tfrac{k}{2}\rfloor \leq \lceil \tfrac{k}{2}\rceil \leq k-1$.
    Hence we may apply the induction hypothesis on $C_1$ and $C_2$, and in particular, there are at most $\bigO(m^{\lceil \tfrac{k}{2}\rceil/2})$ choices for $C_1$ and at most $\bigO(m^{\lfloor \tfrac{k}{2}\rfloor/2})$ choices for $C_2$. 
    Thus, there are at most $\bigO(m^{\lceil \tfrac{k}{2}\rceil/2}\cdot m^{\lfloor \tfrac{k}{2}\rfloor/2}) = \bigO(m^{\lceil \tfrac{k}{2}\rceil/2+ \lfloor \tfrac{k}{2}\rfloor/2}) = \bigO(m^{k/2})$ many choices for $C$.
\end{proof}
{\renewcommand\footnote[1]{}\cliqueDomAlg*}
\begin{proof}
    We first leverage the fact that each dominating set contains a heavy vertex, and there are at most $\bigO(m/n)$ such vertices.
    Moreover, the remaining vertices form a $(k-1)$-clique and by \autoref{obs:k-clique-number}, there are at most $m^{\frac{k-1}{2}}$ choices for the remaining $k-1$ solution vertices and moreover, we can enumerate them in $\bigO(m^{\frac{k-1}{2}})$ time (see e.g. \cite{itai1978trianglelisting}).
    This suggests the following algorithm. 
    Let $R_1$ and $R_2$ denote the sets of all cliques of size $\lfloor \frac{k-1}{2}\rfloor$ and $\lceil \frac{k-1}{2}\rceil$ respectively.
    Note that we can enumerate those cliques in $m^{\frac{1}{2}\lceil \frac{k-1}{2}\rceil}$, as mentioned above. 
    Furthermore, let $\mathcal{H}$ denote the set of all heavy vertices in $G$.
    Construct matrices $A$ and $B$ naturally as follows.
    Index the rows of $A$ with the elements of the set $R_1\times \mathcal{H}$ and the columns by $V(G)$ and set the entry $A[S,v]$ to $1$ if and only if $v\in N[S]$, i.e. $S$ dominates $v$, otherwise $0$.
    Similarly index the columns of $B$ by the elements of $R_2$ (cliques of size $\lceil \frac{k-1}{2}\rceil$) and rows by $V(G)$ and set the entry $B[v, T]$ to $1$ if and only if $T$ dominates $v$ and $0$ otherwise.
    Consider the matrix $C:=\overline{A}\cdot\overline{B}$.
    Iterate over the entries of $C$ and if $C[S,T]=0$, check (in time $\bigO(1)$) if $S\cup T$ forms a clique of size $k$, if it does return YES. If no entry satisfies both of these conditions, return NO.
    Let us first prove the running time of this algorithm.
    As noted above, there are at most $m^{\frac{1}{2}\lfloor \frac{k-1}{2}\rfloor}$ elements in $R_1$ and at most $m^{\frac{1}{2}\lceil \frac{k-1}{2}\rceil}$ elements in $R_2$. 
    Further, there are at most $\bigO(\frac{m}{n})$ elements in $\mathcal H$.
    The running time now follows.
    We conclude the proof by arguing the correctness of the algorithm.
    It is straightforward to verify that the entry $C[S,T] = 0$ if and only if the set $S\cup T$ dominates $G$.
    Hence, if $S\cup T$ additionally forms a clique, we are done. 
    It remains to show that we do not miss any solutions.
    Let $x_1,\dots, x_k$ be a $k$-clique that dominates $G$ and assume $\deg(x_1)\leq  \dots \leq \deg(x_k)$. 
    Clearly, since $x_1,\dots, x_k$ is a dominating set, $x_k$ is a heavy vertex and hence $x_k\in \mathcal{H}$.
    Moreover, since $x_1,\dots, x_k$ form a $k$-clique, the set $\{x_1,\dots, x_{\lfloor \frac{k-1}{2}\rfloor}\}$ is contained in $R_1$, and the set $\{x_{\lfloor \frac{k-1}{2}\rfloor + 1}, \dots, x_{k-1}\}$ is contained in $R_2$.
    It is now straightforward to see that our algorithm will detect this dominating clique, thus concluding our proof. 
\end{proof}
\subsection{Dominating $k$-Independent Set} \label{app-section:indset}
\LemmaIndSetDomTC*
\begin{proof}
    For any set $S$, let $E_{S}$ denote the subset of edges of $G$ with at least one endpoint in $S$. 
    Then by \autoref{lemma:domindset-recursion} combined with the observation that any dominating set contains a heavy node, it holds that
    \[T_k(n,m) \leq \sum_{v\in \mathcal{H}}\big(m + T_{k-1}(n - |N[v]|,m-|E_{N[v]}|)\big).\]
    Where the summand $m$ is the upper bound on the running time needed to create a copy of $G$ without $N[v]$.
    We first notice that $\sum_{v\in \mathcal{H}}m\leq \frac{m^2}{n}$, and since $k\geq 3$, this is at most $T_k(n,m)$, unless $k$-OV hypothesis fails. 
    A simple analysis of this running time shows that it is monotonically increasing in $m$, thus 
    \[\sum_{v\in \mathcal{H}}T_{k-1}(n - |N[v]|,m-|E_{N[v]}|) \leq \sum_{v\in \mathcal{H}}T_{k-1}(n - |N[v]|,m).\]
\end{proof}
\subsection{Dominating $k$-Induced Matching} \label{app-section:matching}
\ThInducedMatchingAlgorithm*
\begin{proof}
    We exploit the observation that every induced matching of size $k$ contains exactly $\frac{k}{2}$ edges that uniquely determine it.
    
    We let $\mathcal{S}$ be the set containing all subsets of $E(G)$ of size $\lceil\frac{k}{4}\rceil$ and $\mathcal{T}$ be the set containing all subsets of $E(G)$ of size $\lfloor \frac{k}{4}\rfloor$.
    Let $A$ be an $m^{\lceil \frac{k}{4}\rceil}\times n$ matrix, whose rows correspond to the elements of $\mathcal{S}$ and columns to vertices of $G$ and similarly, let $B$ be an $n\times m^{\lfloor \frac{k}{4}\rfloor}$ matrix whose rows correspond to vertices of $G$ and columns to the elements of $\mathcal{T}$.
    Set the entries of $A$ and $B$ naturally, in particular let $A[S,v] = 1$ if and only if there exists and edge $(u,w)\in S$ such that the set of its endpoints $\{u,w\}$ dominates $v$.
    We set the entries of $B$ similarly.
    \begin{claim}
        Let $C:=\overline A\cdot \overline B$.
        Then $C[S,T] = 0$ if and only if the endpoints of $S\cup T$ dominate $G$
    \end{claim}
    \begin{proof}
        It is straightforward to observe that $\overline A[S,v] = \overline B[v,T] = 1$ if and only if $S\cup T$ does not dominate $v$. 
        Moreover, since both $A$ and $B$ are $0/1$-matrices, $C[S,T] = 0$ if and only if for each $v$, either $\overline A[S,v]$ or $\overline B[v,T]$ are $0$.
    \end{proof}
    This suggests the following algorithm.
    Construct matrices $A$ and $B$ and compute the product $C=\overline{A}\overline{B}$.
    For each pair $S\in \mathcal{S},T\in \mathcal{T}$, we check two things, first that the entry $C[S,T]=0$ and finally that $S\cup T$ induce a perfect matching on $k$ vertices.
    Note that we can check both of those conditions in constant time.
    If both conditions are satisfied for any pair $S,T$, we return YES, otherwise NO.

    It is now straightforward to verify the correctness and the running time of the algorithm.
\end{proof}
\ThInducedMatchingLB*
\begin{proof}
    To prove this lower bound, we reduce from $k$-Orthogonal Vectors.
    Let $A_1,\dots, A_k\subseteq \{0,1\}^d$ be sets of $d$-dimensional binary vectors.
    Moreover, let $|A_i| = \frac{m}{n}$ if $i$ is even and $n$ otherwise.
    We construct a graph $G$ with $\bigO(n)$ vertices and $\bigO(m)$ edges that has a Induced Matching Dominating Set of size $k$ if and only if $A_1,\dots, A_k$ is a YES instance of $k$-Orthogonal Vectors.

    Let $G = (V,E)$, such that $V = X_1\cup\dots \cup X_k\cup D$, where each vertex $x_i\in X_i$ correspond to a unique vector $a_i\in A_i$ and $D = [d]$  corresponds to the dimensions of the vectors.
    Add an edge between each $u\in X_i$ and $v\in X_{i+1}$ for each odd $i$.
    Finally, add an edge between $x_i\in X_i$ and $t\in D$ if and only if the vector $a_i\in A_i$ corresponding to $x_i$ satisfies $a_i[t] = 0$.

    Without loss of generality, we can assume that for each $i\in [k]$ there are $k+1$ coordinates $t_1,\dots, t_{k+1}$ such that $a[t_j] = 0$ for every $a\in A_i$ and $a[t_j] = 1$ for every $a\not \in A_i$.
    We call these coordinates \emph{special coordinates of $A_i$}.
    Notice that this property implies the following.
    \begin{claim}
        If there are vertices $x_1,\dots, x_k$ that dominate $G$, then (w.l.o.g.) $x_1\in X_1,\dots, x_k\in X_k$. 
    \end{claim}
    \begin{proof}
        Assume for contradiction that there is an index $i\in [k]$ such that no $x_j$ is contained in $X_i$. 
        Consider any vertex $t$ corresponding to a special coordinate of $A_i$. 
        By construction, it is not adjacent to any $x_j$, hence the only way it can be dominated is if $t = x_j$ for some $j\in [k]$.
        However, since there are $k+1$ special coordinates for $X_i$, this leaves at least one of the corresponding vertices undominated, contradicting that the vertices $x_1,\dots, x_k$ dominate $G$.
    \end{proof}
    Assume now that vertices $x_1\in X_1,\dots, x_k\in X_k$ dominate $G$.
    We claim that the corresponding vectors $a_1\in A_1, \dots a_k\in A_k$ satisfy the orthogonality condition.
    Indeed, for each vertex $t\in D$, there is a vertex $x_i$ that is adjacent to $t$, hence the corresponding vector $a_i$ satisfies $a_i[t] = 0$.
    
    Conversely, if $a_1\in A_1,\dots, a_k\in A_k$ are orthogonal, we claim that the corresponding vertices $x_1\in X_1,\dots, x_k\in X_k$ form an Induced Matching Dominating Set.
    First observe that in $G[\{x_1,\dots, x_k\}]$, each vertex has degree $1$ (for odd $i$, the vertex $x_i$ is only adjacent to $x_{i+1}$ and for even $i$, to $x_{i-1}$), so the vertices induce a perfect matching.
    It remains to verify that they dominate $G$.
    To this end consider first an arbitrary vertex $t\in D$. 
    There exists at least one $i\in [k]$, such that $a_i[t]= 0$, and thus the corresponding vertex $x_i$ is adjacent to $t$.
    On the other hand, consider some vertex $y\in X_j$ for some $j\in [k]$.
    If $j$ is odd, then $y$ is adjacent to $x_{j+1}$, otherwise to $x_{j-1}$.
    We have thus exhausted all vertices, and can conclude that $x_1,\dots, x_k$ is indeed an Induced Matching Dominating Set.

    Now observe that $G$ has $\bigO(n)$ vertices and $\bigO(m+dn)$ edges (assuming $d=\log^2(n)$, this is near linear) and can be constructed in $\bigO(m+nd)$ time.
    Assume that there is an algorithm solving $k$-Induced Matching Dominating Set in time $O(m^{\frac{k}{2}-\varepsilon})$.
    Then we can decide if $A_1,\dots, A_k$ as above is a YES-instance by constructing $G$ in time $\bigO(m+dn)$ and solving $k$-Induced Matching Dominating Set on $G$ in time $O(m^{\frac{k}{2}-\varepsilon})$, yielding a solution in $\bigO(m^{\frac{k}{2}-\varepsilon}) = \bigO\big((|A_1|\cdot\dots\cdot |A_k|)^{1-\varepsilon}\big)$, thus refuting the $k$-OV hypothesis.
\end{proof}
\end{document}